\documentclass[a4paper,onecolumn,11pt,accepted=2021-11-11]{quantumarticle}
\pdfoutput=1
\usepackage[utf8]{inputenc}
\usepackage[english]{babel}
\usepackage[T1]{fontenc}
\usepackage{amsmath}
\usepackage{amsthm}
\usepackage{amsfonts}
\usepackage{hyperref}

\usepackage{tikz}
\usepackage{lipsum}

\newtheorem{theorem}{Theorem}
\newtheorem{lem}{Lemma}

\newcommand{\bb}{\mathbb}

\newcommand{\hket}[1]{\ket{\hat{#1}}}
\newcommand{\cat}{\text{cat}}

\newcommand{\ket}[1]{\ensuremath{\lvert #1 \rangle}}
\newcommand{\bra}[1]{\ensuremath{\langle #1 \rvert}}
\newcommand{\csize}{c}

\newcommand\blfootnote[1]{%
  \begingroup
  \renewcommand\thefootnote{}\footnote{#1}%
  \addtocounter{footnote}{-1}%
  \endgroup
}

\begin{document}
\title{Improved upper bounds on the stabilizer rank of magic states}

\author{Hammam Qassim }
\affiliation{Institute for Quantum Computing, University of Waterloo, Ontario}
\affiliation{Department of Physics and Astronomy, University of Waterloo, Ontario}
\affiliation{Keysight Technologies Canada, Inc.}
\author{Hakop Pashayan}
\affiliation{Institute for Quantum Computing, University of Waterloo, Ontario}
\affiliation{Department of Combinatorics and Optimization, University of Waterloo, Ontario}
\affiliation{Perimeter Institute for Theoretical Physics, Waterloo, Ontario}

\author{David Gosset}
\affiliation{Institute for Quantum Computing, University of Waterloo, Ontario}
\affiliation{Department of Combinatorics and Optimization, University of Waterloo, Ontario}
\affiliation{Perimeter Institute for Theoretical Physics, Waterloo, Ontario}

\maketitle
\begin{abstract}
In this work we improve the runtime of recent classical algorithms for strong simulation of quantum circuits composed of Clifford and T gates. The improvement is obtained by establishing a new upper bound on the stabilizer rank of $m$ copies of the magic state $\ket{T}=\sqrt{2}^{-1}(\ket{0}+e^{i\pi/4}\ket{1})$ in the limit of large $m$. In particular, we show that $|T\rangle^{\otimes m}$ can be exactly expressed as a superposition of at most $O(2^{\alpha m})$ stabilizer states, where $\alpha\leq 0.3963$, improving on the best previously known bound $\alpha \leq 0.463$. 
This furnishes, via known techniques, a classical algorithm which approximates output probabilities of an $n$-qubit `Clifford + T' circuit $U$ with $m$ uses of the T gate to within a given inverse polynomial relative error using a runtime $\mathrm{poly}(n,m)2^{\alpha m}$. We also provide improved upper bounds on the stabilizer rank of symmetric product states $\ket{\psi}^{\otimes m}$ more generally; as a consequence we obtain a strong simulation algorithm for circuits consisting of Clifford gates and $m$ instances of any (fixed) single-qubit $Z$-rotation gate with runtime  $\text{poly}(n,m) 2^{m/2}$. We suggest a method to further improve the upper bounds by constructing linear codes with certain properties. 
\end{abstract}

\maketitle
\blfootnote{A preliminary version of our results was reported in the first author's Ph.D thesis~\cite{qassim}.}

Recently proposed classical algorithms for simulating quantum circuits dominated by Clifford gates~\cite{garcia2014PhD,bravyi2016trading,bravyi2016improved,bravyi2018simulation} have an exponential scaling only in the number of non-Clifford gates in the circuit.  Such algorithms are based on representing the output state $\ket{\Psi}=U|0^n\rangle$ of an $n$-qubit circuit $U$ as a superposition
\begin{align}\label{stab decomp}
\ket{\Psi} = \sum_{i=1}^k c_i \ket{\Phi_i},
\end{align}
where $\{\ket{\Phi_i}\}$ are $n$-qubit stabilizer states, each of which requires only $O(n^2)$ classical bits to store. Given such a representation we can compute an amplitude $\langle x|U|0^n\rangle$  corresponding to a given basis state $x\in \{0,1\}^n$ by computing a weighted sum of the $k$ amplitudes $\langle x|\Phi_i\rangle$, each of which can be computed in polynomial time via the stabilizer formalism. It was shown in Ref.~\cite{bravyi2016improved} that we can also approximate a measurement probability $p(y)=\| |y\rangle\langle y|_{\mathrm{out}} |\Psi\rangle\|^2$ for an output register containing $w\leq n$ qubits and $y\in \{0,1\}^w$ to within a given inverse polynomial relative error using a runtime linear in $k$ and polynomial in the number of qubits. Simulations of this type are advantageous if a stabilizer decomposition Eq.~\eqref{stab decomp} can be found with $k \ll 2^n$. The \textit{stabilizer rank} $\chi(\Psi)$ is defined as the minimum possible $k$ over all decompositions of the form Eq.~\eqref{stab decomp}. Let us suppose that $\Psi$ is the output state of a circuit $U$  composed of Clifford gates and at most $m$ T gates.   In this case it is shown in Ref.~\cite{bravyi2016trading} that $\chi(\Psi)\leq \chi(T^{\otimes m})$ where $\ket{T} = 2^{-1/2}(\ket{0}+e^{i\pi/4} \ket{1})$ is a single-qubit magic state.  In this way the stabilizer rank of magic states $\chi(T^{\otimes m})$ directly relates to the complexity of computing high-precision (exact, or inverse polynomial relative error) approximations of amplitudes and probabilities of Clifford+T circuits.  Though it is not our focus in this paper, we note that one may also consider a less stringent weak simulation task where the goal is to approximately sample from the output distribution of the quantum circuit to within a given inverse polynomial total variation distance.  Ref.~\cite{bravyi2016improved} defines an approximate stabilizer rank and shows how it is similarly related to the runtime of weak simulation.

Computing the stabilizer rank of a given $n$-qubit state appears to be intractable. The main difficulty is that the number of stabilizer states grows as $2^{\Theta(n^2)}$, and searching for a decomposition of the form Eq.~\eqref{stab decomp} by exhaustively considering $k$-tuples of $n$-qubit stabilizer states is out of the question, even for very small $n$ and $k$. As discussed above, we are interested in upper bounding the stabilizer rank of magic states $\ket{T}^{\otimes m}$ as this captures the classical simulation cost for Clifford+T circuits with $m$ T gates. The fact that this is a product state has been exploited in previous work to bound its stabilizer rank.  Indeed, for any two states $\phi_1, \phi_2$ we have the trivial sub-multiplicativity property $\chi(\phi_1\otimes \phi_2)\leq \chi(\phi_1)\chi(\phi_2)$. Previous works have upper bounded $\chi(T^{\otimes m})$ by constructing low-rank decompositions of $\ket{T}^{\otimes \csize}$ for small $\csize$ and then using this sub-multiplicativity. For example, the fact that $\ket{T}^{\otimes 2}$ can be decomposed in terms of two stabilizer states;
\begin{align}
\ket{T}^{\otimes 2} = \frac{1}{2}\Big(\ket{00}+ i \ket{11} \Big) + \frac{e^{i \pi/4}}{2}\Big(\ket{01}+\ket{10} \Big),  
\end{align}
implies $\chi(T^{\otimes m}) = O(2^{m/2})$. This bound was improved in Ref.~\cite{bravyi2016trading} by constructing stabilizer decompositions of $\ket{T}^{\otimes c}$ for $c\leq 6$ using a heuristic simulated-annealing based computer search. The best decomposition found in Ref.~\cite{bravyi2016trading} uses 7 stabilizer states to decompose $\ket{T}^{\otimes 6}$, implying that $\chi(T^{\otimes 6}) \leq 7$ and therefore $\chi(T^{\otimes m}) \leq O(2^{\alpha m})$ where $\alpha=\log_2(7)/6 \leq 0.468$. The equality $\chi(T^{\otimes 6}) = 7$ was conjectured to hold, suggesting that further improvement requires low-rank decompositions of 7 or more copies of $\ket{T}$. Ref.~\cite{kocia2020improved} shows that  $\chi(T^{\otimes 12})\leq 47$, which further reduces the upper bound on the exponent $\alpha$ to roughly $0.463$. Though it is not our focus in this work, we note that existing \textit{lower} bounds on the stabilizer rank of magic states are unsatisfying: the best known unconditional lower bound is $\chi(T^{\otimes m})=\Omega(m)$~\cite{peleg2021lower}. 
On the other hand this quantity must increase exponentially with $m$ unless \#P-complete problems can be solved in sub-exponential time on a classical computer (see, e.g.,  Refs.~\cite{huang2019explicit, morimae2019fine, huang2019tgates}).
\begin{figure}[t]
\centering
\includegraphics[scale=0.045]{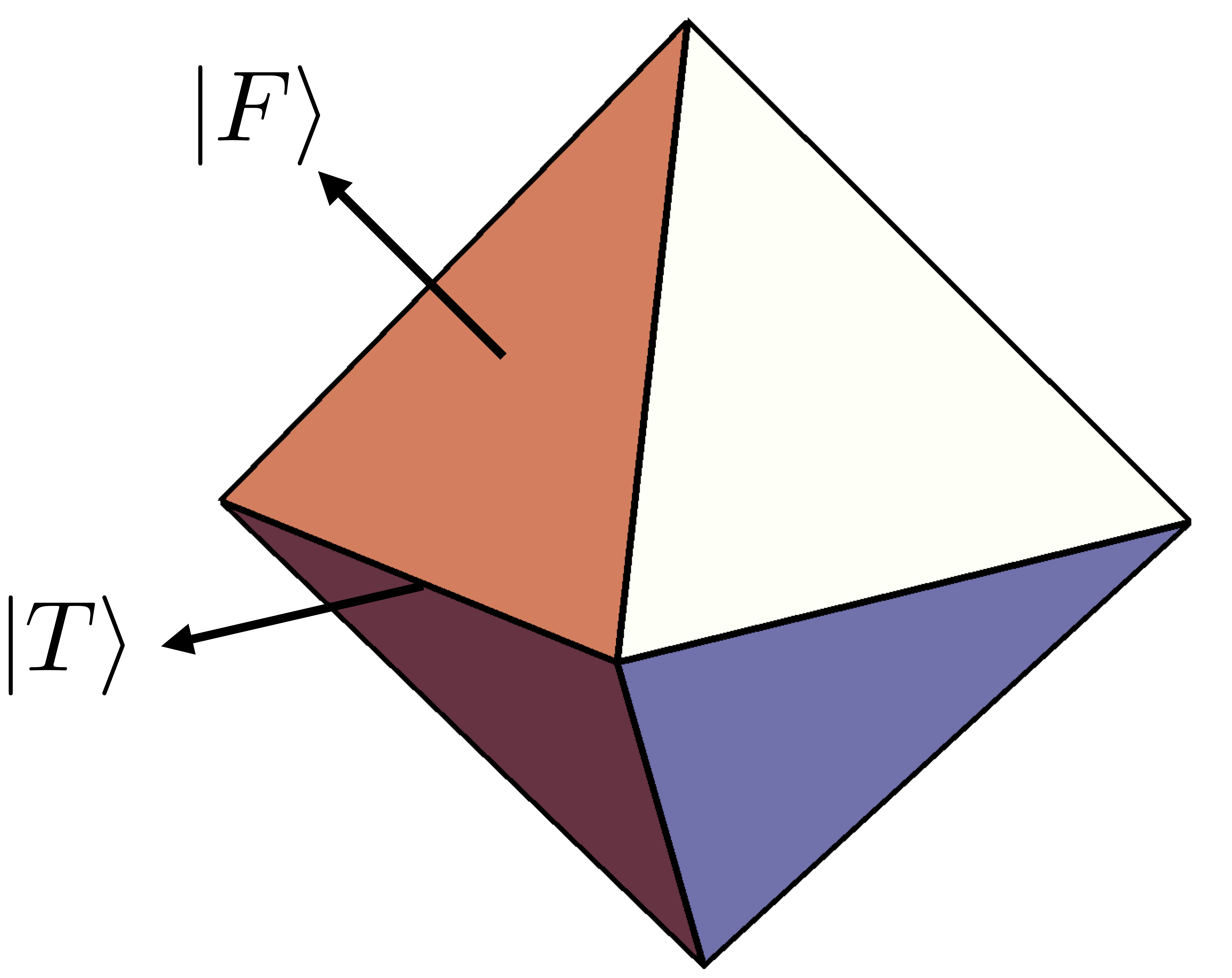}
\caption{The stabilizer octahedron, whose vertices are the six single qubit stabilizer states. The two types of magic states, $\ket{T}, \ket{F}$, correspond to the edges and faces of this octahedron, respectively~\cite{bravyi2005universal}. It was conjectured in Ref.~\cite{bravyi2016trading} that $\chi(T^{\otimes m}) = \chi(F^{\otimes m})$.}
\label{fig: magic states}
\end{figure} 

In this paper, we describe a new method of constructing low-rank stabilizer decompositions of $\ket{T}^{\otimes m}$ and obtain the improved upper bound $\chi(T^{\otimes m})\leq O(2^{\alpha m})$ with $\alpha = \log_2(3)/4 \leq 0.3963$. Unlike previous works, our strategy does not involve bounding $\chi(T^{\otimes \csize})$ for some small $\csize$ and then using sub-multiplicativity. Rather, we construct low-rank stabilizer decompositions of the state $\ket{T}^{\otimes m}$ by contracting certain entangled cat states in the magic basis.  Along the way, we establish that $\chi(T^{\otimes 6}) \leq 6$, disproving the conjecture that it is equal to $7$. We generalize our upper bound to the other species of magic states (see Fig.~\ref{fig: magic states}), as well as to other symmetric product states. Finally, we discuss a generalization of our method which is based on a family of quantum states corresponding to linear codes.
\paragraph{Magic cat states} For each $m\geq 1$,  define the $m$-qubit cat state in the magic basis 
\begin{equation}
\ket{\cat_m} \equiv 2^{-1/2}(\ket{T}^{\otimes m} + \ket{T^\perp}^{\otimes m}),
\label{eq:magiccat}
\end{equation}
where $\ket{T^\perp} \equiv Z\ket{T}$. Similar to the states $\ket{T}^{\otimes m}$, the stabilizer rank of $\ket{\cat_m}$ is non-decreasing in $m$, i.e. $\chi(\cat_m) \geq \chi(\cat_{m-1})$, which follows by noting that $\ket{\cat_{m-1}} \propto (\bra{0} \otimes I) \ket{\cat_m}$. Furthermore it is easy to see that the stabilizer rank of this magic cat state is within a factor of $2$ of the stabilizer rank of $\ket{T}^{\otimes m}$. In particular, we have
\begin{equation}
\frac{\chi(T^{\otimes m})}{2} \leq \chi(\cat_m)\leq \chi(T^{\otimes m})
\label{eq:catt}
\end{equation}
The upper bound follows because $\ket{\cat_m} =\sqrt{2}^{-1} (I + Z^{\otimes m})\ket{T}^{\otimes m}$, where $2^{-1}(I + Z^{\otimes m})$ is a stabilizer projector which does not increase the stabilizer rank. The lower bound in Eq.~\eqref{eq:catt} follows by using the fact that $\ket{T}\bra{T} = \frac{1}{2}(I + A)$ for a single-qubit Clifford $A$\footnote{Namely $ A= e^{-i\pi/4} S X$ where $X$ is the usual Pauli matrix and $S = \text{diag}(1,i)$.}, and therefore $\ket{T}^{\otimes m} = \sqrt{2} (\ket{T}\bra{T}\otimes I)\ket{\cat_m} \propto \ket{\cat_m}+(A \otimes I)\ket{\cat_m}$ where the right-hand-side has stabilizer rank of at most $2\chi(\cat_m)$.

We shall take advantage of the fact that the stabilizer rank of the $m$-qubit magic cat state is less than that of $|T\rangle^{\otimes m}$ for small $m$. For $m=2,6$ we find the following stabilizer decompositions
\begin{align}\label{T cat decomps}
\ket{\cat_2}&= 2^{-1/2}(\ket{0}^{\otimes 2}+ i \ket{1}^{\otimes 2}),\nonumber\\
\ket{\cat_6} &= 2^{-3/2}(\ket{0}^{\otimes 6} 
- i\ket{1}^{\otimes 6})+  2^{-1/2} e^{3 i \pi/4} (\ket{E_6} +  i \ket{K_6}),
\end{align} 
where we defined the stabilizer states\footnote{A state is a stabilizer state if and only if its expansion in the computational basis is of the form $|A|^{-1/2}\sum_{x \in A} i^{l(x)} (-1)^{q(x)} \ket{x}$, where A is an affine subspace of $\mathbb{F}_2^n$ and $l:\mathbb{F}_2^n \rightarrow \mathbb{F}_2$ and $q:\mathbb{F}_2^n \rightarrow \mathbb{F}_2$ are linear and quadratic functions, respectively~\cite{Dehaene2003Clifford, vdnest2010classical}.}
$$\ket{E_m} \equiv 2^{-(m-1)/2}\sum_{\underset{|x| \text{ even}}{x \in \{ 0,1\}^m}: } \ket{x}, \quad \ket{K_m} \equiv \prod_{1\leq i<j\leq m} CZ_{ij}\ket{E_m}.$$
From these decompositions we see that 
\begin{equation}
\chi(\cat_2) = 1 \quad \text{ and }\quad  \chi(\cat_6) \leq 3.
\label{eq:cat26}
\end{equation}

The latter inequality already disproves the conjecture that $\chi(T^{\otimes 6}) = 7$. Indeed, using Eq.~\eqref{eq:catt} we infer 
\begin{equation}
\chi(T^{\otimes 6}) \leq 2 \chi(\cat_6) \leq 6.
\label{eq:t6}
\end{equation}
\begin{figure}[t]
\centering
\includegraphics[scale=0.12]{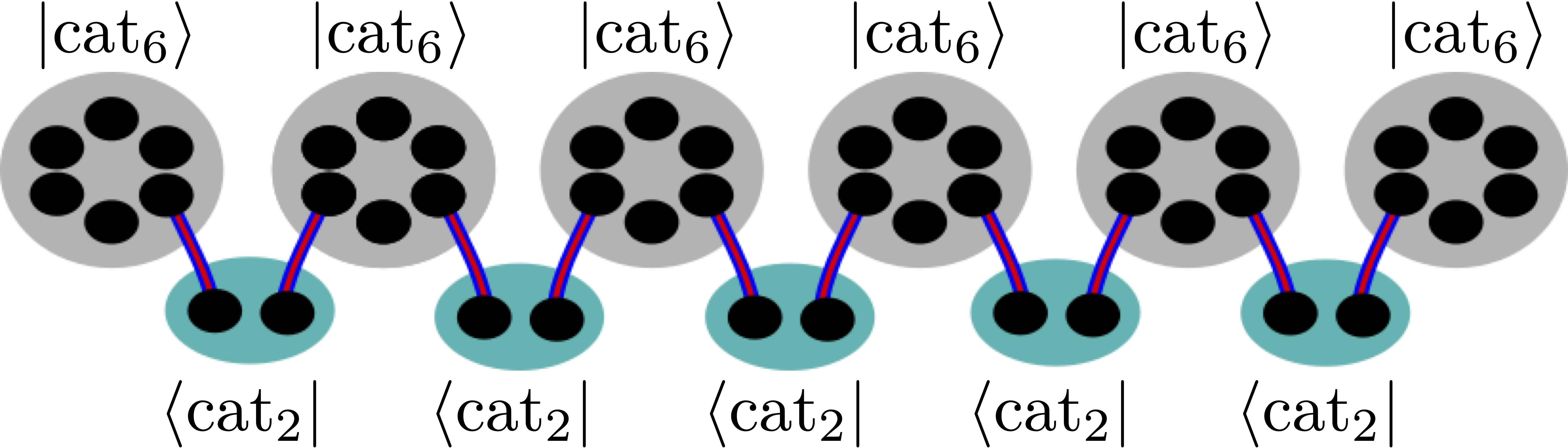}
\caption{The chain used in the proof of Theorem \ref{res magic}. Here $\ell = 6$ and the state after the contraction is $\ket{\cat_{26}}$.}
\label{fig: cat chain 1}
\end{figure} 
Following the strategy from Ref.~\cite{bravyi2016trading} we may then use sub-multiplicativity of the stabilizer rank to obtain $\chi(T^{\otimes m}) = O(2^{\alpha m})$ with $\alpha = \log_2(6)/6\leq 0.4308$.

But we can do better by using Eqs.~\eqref{eq:cat26} in a different way. 
Suppose we act with the stabilizer state $\bra{\cat_2}$ on one qubit from each tensor factor in $\ket{\cat_6}\ket{\cat_6}$.  The resulting state is 
\[
\langle \cat_2|_{6,7} ~\ket{\cat_6}\ket{\cat_6}\propto \ket{\cat_{10}}.
\]
By expanding each state $|\mathrm{cat}_6\rangle$ as a superposition of $3$ stabilizer states and using the fact that $\langle \cat_2|$ is a stabilizer state, we obtain a stabilizer decomposition of $\ket{\cat_{10}}$ with only 9 terms. Combining this with Eq.~\eqref{eq:catt} we infer $\chi(T^{\otimes 10})\leq 18$ and therefore $\chi(T^{\otimes m}) = O(2^{\alpha m})$ with $\alpha = \log_2(18)/10\leq 0.417$.

We can obtain further improvements by contracting many stabilizer $\ket{\cat_2}$ states with many copies of $\ket{\cat_6}$ to construct stabilizer decompositions of larger magic cats $|\text{cat}_m\rangle$. In the limit of large $m$ this gives the exponent $\alpha = \log_2(3)/4 \leq 0.3963$. To see this, in Fig.~\ref{fig: cat chain 1} we start with a chain of length $\ell$ containing a $\ket{\cat_6}$ state at each site, and we contract nearest neighbours using $\bra{\cat_2}$. After the contraction the state of the remaining qubits  is (up to normalization) $\ket{\cat_{4\ell + 2}}$.
Therefore a stabilizer decomposition of $\ket{\cat_{4\ell + 2}}$ exists using $3^\ell$ stabilizer states. Using Eq.~\eqref{eq:catt} we have
\[
\frac{\chi(T^{\otimes (4\ell+2)})}{2}\leq \chi(\cat_{4\ell+2})\leq 3^{\ell} \qquad \ell\geq 2.
\]
Therefore one can decompose $\ket{T}^{\otimes m}$ using $O(2^{\alpha m})$ stabilizer states, where $\alpha = \log_2(3)/4$. 

The above argument can be readily generalized to the other (Clifford-inequivalent) single-qubit magic state $\ket{F} = \cos(\beta)\ket{0} + e^{i \pi/4} \sin(\beta) \ket{1}$ where $\cos(2\beta) = \frac{1}{\sqrt{3}}$, $0 \leq \beta \leq \pi$ \cite{bravyi2005universal}. Defining $\ket{F^\perp} = \sin(\beta)\ket{0} - e^{i \pi/4} \cos(\beta) \ket{1}$ and $\ket{\cat_m(F)} = 2^{-1/2}(\ket{F}^{\otimes m} + \ket{F^\perp}^{\otimes m})$ we find
\begin{align}\label{F cat decomps}
\ket{\cat_2(F)}&= 2^{-1/2}(\ket{0}^{\otimes 2}+ i \ket{1}^{\otimes 2}),\nonumber\\
\ket{\cat_6(F)} &= (2/3) \ket{\psi_1}
+ e^{3i\pi/4}(2/3) \ket{\psi_2} - e^{i\pi/4}(2/3) \ket{\psi_3},
\end{align} 
where $\ket{\psi_i}$ are the stabilizer states
\begin{align}
\ket{\psi_1} &= 2^{-1/2}(\ket{0}^{\otimes 6} - i \ket{1}^{\otimes 6}),\nonumber\\
\ket{\psi_2} &= 2^{-3} \sum_{x \in \{0,1 \}^6} (-i)^{|x| \text{ mod }2} \ket{x},\nonumber\\
\ket{\psi_3} &= 2^{-3} \sum_{x \in \{0,1 \}^6} (-1)^{|x|(|x|+1)/2} \ket{x}.
\end{align}
Therefore $\chi(\cat_2(F)) = 1$ and $\chi(\cat_6(F)) \leq 3$, and so the above method used to construct stabilizer decompositions of $\ket{T}^{\otimes m}$ can be applied directly to $\ket{F}^{\otimes m}$. We summarize the preceding discussion in the following theorem.
 
 \begin{theorem}\label{res magic}
For $\ket{\psi} \in \{\ket{T}, \ket{F}\}$ we have
$
\chi(\psi^{\otimes m}) = O ( 2^{\alpha m } ),
$
where $\alpha = \log_2(3)/4 \leq 0.3963$.
\end{theorem}

\begin{table}[t]
\centering
\begin{tabular}{|c|c|c|c|c|c|c|c|}
\hline
m                 & 2 & 3 & 4       & 5       & 6       & 7        & 8        \\ \hline
$\chi(T^{\otimes m})$ & 2 & 3 & $\leq$4 & $\leq$6 & $\leq$6 & $\leq$12 & $\leq$12 \\ \hline
$\chi(\cat_m)$        & 1 & 2 & 2       & 3       & 3       & $\leq$6  & $\leq$6  \\ \hline
\end{tabular}
\caption{Known values of the stabilizer rank of $\ket{T}^{\otimes m}$ and $\ket{\cat_m}$ for small $m$. The top row includes the bounds from Refs.~\cite{bravyi2016trading, bravyi2018simulation} for $m=2,3,4,5,7$ and the bound obtained from  Eq.~\eqref{eq:t6} for $m=6$. The upper bound on $\chi(T^{\otimes 8})$ follows from Eq.~\eqref{eq:catt} and $\chi(\cat_8)\leq 6$.  Further details of this and the other bounds listed in the bottom row are provided in the Appendix.}
\label{tab}
\end{table}

As we have already discussed, this provides the fastest currently known classical algorithm for estimating amplitudes or output probabilities of Clifford+T circuits as a function of the number of T gates. It is an interesting question whether or not Theorem \ref{res magic} can be used to accelerate classical algorithms for other computational problems which are less directly tied to quantum computing\footnote{For example, Ref. \cite{huang2019tgates} shows that if $\chi(T^{\otimes m})\leq 2^{\alpha_0 m}$ where $\alpha_0=2.2451\cdot 10^{-8}$ then one obtains a classical algorithm for 3-SAT which outperforms existing methods.}.

A natural question is whether the bound $\log_2(3)/4$ can be improved using our techniques.
As we show in the Appendix, the stabilizer rank of $|\text{cat}_6\rangle$ is exactly $3$, so an improvement using the above method requires the use of other cat states.
In Table~\ref{tab} we list the known values of $\chi(\cat_m)$ for small $m$, none of which provide a better upper bound than the one in Theorem \ref{res magic}. Ultimately it might be possible to improve the bound by finding low rank decompositions of $\ket{\cat_m}$ for larger $m$. Doing so may be slightly easier than finding decompositions of $\ket{T}^{\otimes m}$ since we always have $\chi(\cat_m)\leq \chi(T^{\otimes m})$. Note that the exponent $\alpha = \log_2(3)/4$ in Theorem \ref{res magic} is only attained asymptotically, i.e. in the limit of an infinitely long chain.

Let us now apply a similar strategy to equatorial states, i.e. state of the form $\ket{R_\theta} = 2^{-1/2}(\ket{0} + e^{i\theta} \ket{1})$. Defining $\ket{R_\theta^\perp} = Z\ket{R_\theta}$ and $\ket{\cat_m(R_\theta)} = 2^{-1/2}(\ket{R_\theta}^{\otimes m} + \ket{R_\theta^\perp}^{\otimes m})$, one finds the stabilizer decompositions

\begin{align}
\ket{\cat_2(R_\theta)} = &2^{-1/2}(\ket{0}^{\otimes 2}+ e^{2 i \theta} \ket{1}^{\otimes 2}), \nonumber\\
\ket{\cat_6(R_\theta)} = 
 &2^{-5/2} \left[(1 - e^{4 i \theta}) \ket{0}^{\otimes 6} 
+(e^{6 i \theta} - e^{2 i \theta}) \ket{1}^{\otimes 6}\right] \nonumber\\
+ &e^{3 i \theta}  \left[  \cos \theta \ket{E_6}
+  i\sin \theta \ket{K_6}\right].
\end{align}
In particular for $\theta \notin (\pi/4) \mathbb{Z}$ we have $\chi(\cat_2(R_\theta)) =  2$ and $\chi(\cat_6(R_\theta)) \leq 4$. Unlike the previous case the two-qubit cat state $\ket{\cat_2(R_\theta)}$ is generally not a stabilizer state. Nevertheless we prove the following theorem using a slightly different contraction of the one-dimensional chain of six-qubit cat states.
\begin{figure}[t]
\centering
\hspace*{-0.5cm}
\includegraphics[scale=0.09]{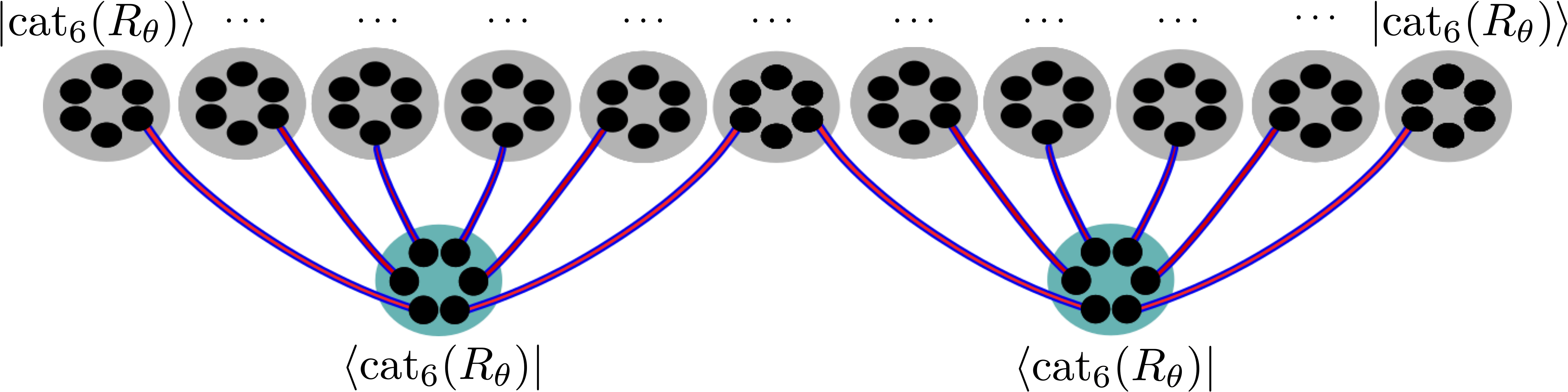}
\caption{The chain used in the proof of Theorem \ref{res equatorial}. Here $t = 2$, $\ell = 11$, and the state after the contraction is $\ket{\cat_{54}(R_\theta)}$.}
\label{fig: cat chain 2}
\end{figure} 
\begin{theorem}\label{res equatorial}
For every angle $\theta$ we have
$
\chi(R_{\theta}^{\otimes m}) = O ( 2^{m/2} ).
$
\end{theorem}
\begin{proof}
We use a chain as in Fig.~\ref{fig: cat chain 2}, where the length (that is, the number of upper/gray circles) is given by $\ell = 5 t +1$ for some $t \in \bb{N}$ (which is the number of lower/turquoise circles). The state after the contraction is $\ket{\cat_{24t + 6}(R_\theta)}$. By decomposing each $\ket{\cat_6(R_\theta)}$ in terms of $4$ stabilizer states we obtain a decomposition of $\ket{\cat_{24t + 6}(R_\theta)}$ using $4^{6t + 1}$ stabilizer states. Acting on the first qubit by $\ket{R_\theta}\bra{R_\theta}$ gives the state $\ket{R_\theta}^{\otimes (24t + 6)}$, and increases the number of stabilizer terms in the decomposition by at most a factor of four (since $\ket{R_\theta}\bra{R_\theta}$ can be written as sum of 4 Paulis). Hence $\chi(R_\theta^{\otimes (24t + 6)}) \leq 2^{12 t + 4}$, which implies the statement.
\end{proof}

The stabilizer decompositions in Theorem \ref{res equatorial} can be used to simulate quantum circuits consisting of Clifford gates and Z-rotations by a fixed angle $\theta$. Indeed,  if $U$ consists of Clifford gates and $m$ instances of a diagonal gate $\exp\left(i \theta Z/2\right)$ then one can write the output state $U\ket{0^{\otimes n}}$ in the form
\begin{align}
U \ket{0}^{\otimes n} = 2^{m/2} \left( I \otimes \bra{0}^{\otimes m} \right) U' \ket{0^{\otimes n}}\ket{R_\theta^{\otimes m}}, 
\end{align}
where $U'$ is an $n+m$-qubit Clifford circuit obtained by replacing each instance of $\exp\left(i \theta Z/2\right)$ with an injection gadget that consumes an ancilla qubit initialized in the state $\ket{R_\theta}$ (see for example, Section 2.3.1 of Ref.~\cite{bravyi2018simulation}). Using the stabilizer decomposition of $\ket{R_\theta}^{\otimes m}$ from Theorem \ref{res equatorial}, we obtain an expression for the output state of the circuit as a superposition of at most $O(2^{m/2})$ stabilizer states. We can then directly compute amplitudes $\langle x|U|0^n\rangle$ using runtime $O(2^{m/2} \text{poly}(n,m))$. We can also estimate marginal probabilities to within an inverse polynomial relative error using the same asymptotic runtime via the ``norm estimation" technique described in Ref. \cite{bravyi2016improved}.

It is in fact possible to extend the exponential scaling in Theorem \ref{res equatorial} to arbitrary symmetric product states. This is a consequence of the following property of the symmetric subspace of $m$ qubits.

\begin{lem}[Theorem 7.5 in \cite{watrous2018theory}]\label{watrous lemma}
Let $S(A) \subseteq \bb{C}^{\otimes 2}$ be the set of vectors with entries in $A \subseteq \bb{C}$. If $|A| \geq m+1$ then the set
$$\{ \ket{v}^{\otimes m} : \ket{v} \in S(A) \}$$ spans the symmetric subspace of $m$ qubits.
\end{lem}
In light of Theorem \ref{res equatorial} the above lemma implies the following bound.
\begin{theorem}\label{res general}
For every single-qubit state $\ket{\psi}$ we have
$
\chi(\psi^{\otimes m}) = O\left( (m+1)2^{m/2} \right).
$
\end{theorem}
\begin{proof}
Let $A = \{e^{2\pi i s /m+1}: s =0, \dots, m\}$. By Lemma \ref{watrous lemma}, the set $ \{ \ket{v}^{\otimes m} : \ket{v} \in S(A) \} $ spans the symmetric subspace of $m$ qubits. Since $S(A)$ consists of $m+1$ distinct equatorial states up to global phases and normalization, the state $\ket{\psi}^{\otimes m}$ can be written in terms of $m+1$ tensor powers of equatorial states;
\begin{align}
\ket{\psi}^{\otimes m} = \sum_{i = 1}^{m+1} c_i \ket{R_{\theta_i}}^{\otimes m}.
\end{align}
Combining this with Theorem \ref{res equatorial} proves the statement.
\end{proof}

\paragraph{Magic code states}
We can extend the method described above by replacing the magic cat state by a more general entangled state of the form
\begin{align}
\ket{\hat{L}} = 2^{-k/2}\sum_{x \in L} \hket{x},
\end{align}
where $\hket{0} = \ket{T}, \hket{1} = \ket{T^\perp}$, and $L \subseteq \mathbb{F}_2^m$ is a linear space of dimension $k$. Note that $\ket{\cat_m}$ is recovered by taking $L$ to be the $m$-bit repetition code. As with the cat states, every $\hket{L}$ has stabilizer rank upper bounded by $\chi(T^{\otimes m})$, since $\hket{L}$ can be obtained from $\ket{T}^{\otimes m}$ using only Pauli measurements and postselection
$$
\hket{L} \propto \prod_{i=1}^{k} (I + Z(v^i)) \ket{T}^{\otimes m},
$$
where $Z(v):= Z^{v_1} \otimes \cdots \otimes Z^{v_m}$ and $\{ v^i \} \subseteq \mathbb{F}_2^m$ is any basis of $L$.  
\begin{theorem}\label{code bound thm}
Suppose there exist polynomials $p$ and $q$ and a constant $\gamma > 0$ such that 
\begin{align}\label{eq: thm5}
\Omega(2^{\gamma n}/p(n)) \leq \chi(T^{\otimes n}) \leq O(q(n)2^{\gamma n}).    
\end{align}
Then for every linear code $L\subseteq \mathbb{F}_2^m$ of block length $m$ and dimension $k < m/2$ we have
\begin{align}\label{code bound}
\gamma \leq \frac{\log_2(\chi(\hat{L}))}{m - 2k}.
\end{align}
\end{theorem}
\begin{proof}
Let $G \in \mathbb{F}_2^{k \times m}$ be a generator matrix of $L$, that is, a binary matrix whose rows span $L$. Up to re-ordering the qubits we may assume that $G$ is in the form
$$
G = [I_{k \times k} \;|\; B],
$$
for a binary matrix $B$ of size $k$ by $(m-k)$.
As $L$ is spanned by the rows of $G$, the constraints $x \in L, x_1 = x_2 =  \cdots = x_k = 0$ are uniquely satisfied by $x = 0^m$. We infer that
\begin{align}\label{eq code contraction}
\left( \bra{T}^{\otimes k} \otimes I\right) \ket{\hat{L}} \propto \ket{T}^{\otimes m-k},
\end{align}
which in turn implies $\chi(T^{\otimes m-k}) \leq \chi(T^{\otimes k}) \chi(\hat{L}).$ Repeating the contraction in Eq.~\eqref{eq code contraction} in parallel $\ell$ times gives
$
\chi(T^{\otimes \ell(m-k)}) \leq \chi(T^{\otimes \ell k}) \chi(\hat{L}^{\otimes \ell}). 
$
Letting $\ell$ get large and using Eq.~\eqref{eq: thm5} yields 
$2^{\gamma\ell(m-k)} \leq f(\ell) 2^{\gamma \ell k} \chi(\hat{L}^{\otimes \ell})$, for some polynomial $f(\ell)$. Taking the binary log, rearranging, and using sub-multiplicativity of the stabilizer rank we get
\begin{align}
\gamma \leq \frac{\log_2(f(\ell)\chi(\hat{L}^{\otimes \ell}))}{\ell(m- 2k)} \leq \frac{\log_2(\chi(\hat{L}))}{m - 2k} \;+ \; \frac{\log_2(f(\ell))}{\ell(m-2k)}.
\end{align}
The second term vanishes for large $\ell$ which completes the proof.
\end{proof}
The upper bound from Theorem \ref{res magic} is recovered from Eq.~\eqref{code bound} by setting $L$ to be the six-bit repetition code (so that $\hket{L} = \ket{\cat_6}$). 
More generally, to use  Eq.~\eqref{code bound}  we need an upper bound on $\chi(\hat{L})$ which is significantly better than the trivial upper bound $\chi(\hat{L}) \leq 2^{m- k}$. The latter upper bound is obtained by noting 
\begin{align}\label{code state projection 2}
\hket{L} \propto \left( \prod_{i=1}^{m-k} (I + A(u^i)) \right) \ket{0}^{\otimes m},   
\end{align}
where $\{u^i\}$ is any basis of $L^\perp$, and we define the Clifford operators $A:= e^{-i\pi/4} S X$ and $A(u) = A^{u_1} \otimes \cdots \otimes A^{u_m}$.
(To see this, note that $\ket{0}^{\otimes m} \propto \sum_x \hket{x}$, $A\hket{0}=\hket{0}$, and $A\hket{1} = - \hket{1}$).

Theorem \ref{code bound thm} requires $k < m/2$, which is somewhat restrictive. Relaxing this requirement we can easily find magic code states with low stabilizer rank. For example, every Reed-Muller code $R_{a,b}$ (of degree $a$ on $b$ variables) corresponds to a stabilizer state if the degree is large enough. 
This is true since for any code $L$ we can express $\hket{L}$ in the computational basis as $\hket{L} \propto \sum_{x \in L^\perp}e^{i\pi |x|/4}\ket{x}$, and if $\lceil b/(b-a-1) \rceil \geq 4$ then $|x| = 0 \mod 8$ for all $x \in R_{a,b}^\perp$ \cite{macwilliams1977theory}. In this case we have $\ket{\hat{R}_{a,b}} \propto \sum_{x \in R_{a,b}^{\perp}} \ket{x}$, and hence it is a stabilizer state.
Another example is the magic code state associated with the 8-bit first order Reed-Muller code. Since $R_{1,3}= R_{1,3}^\perp$ and $|x| = 4$ for all $x \in R_{1,3}$ except the all-zero and all-one binary strings, we have the stabilizer decomposition
\begin{align}
\ket{\hat{R}_{1,3}}\propto 2(\ket{0}^{\otimes 8} + \ket{1}^{\otimes 8}) - \sum_{x \in {R_{1,3}}} \ket{x},
\end{align}
hence $\chi(\hat{R}_{1,3}) = 2$ (one can easily show that $\ket{\hat{R}_{1,3}}$ is not a stabilizer state).

In the preceding examples the code dimension $k$ exceeds the threshold $k < m/2$ required by Theorem \ref{code bound thm}, and so one cannot obtain an upper bound on $\chi(T^{\otimes n})$ that way. We leave it as an open question whether the bound in Theorem \ref{res magic} can be improved using this formulation in terms of linear codes, i.e. whether there exists a linear code $L$ of block length $m$ and dimension $k < m/2$ such that $\log_2(\chi(\hat{L}))/(m-2k) < \log_2(3)/4$.
\section*{Acknowledgments}
This research was funded in part by Army Research Office Grant number W911NF-14-1-0103, and by Transformative Quantum Technologies (TQT). HP acknowledges the support of the Natural Sciences and Engineering Research Council of Canada (NSERC) discovery grant RGPIN-2018-05188. Research at Perimeter Institute is supported in part by the Government of Canada through the Department of Innovation, Science and Economic Development Canada and by the Province of Ontario through the Ministry of Colleges and Universities. DG acknowledges the support of the Natural Sciences and Engineering Research Council of Canada through grant number RGPIN-2019-04198, the Canadian Institute for Advanced Research, and IBM Research. We thank Sergey Bravyi, Alex Hu, Anirudh Krishna, and Joel Wallman for helpful discussions related to this work. 

\bibliography{library}
\bibliographystyle{unsrt}

\appendix
\section{Stabilizer rank of small magic cats}
Here we provide further details of the claimed values of $\chi(\cat_m)$ listed in Table \ref{tab}.  

First let us consider the upper bounds. We have already shown $\chi(\cat_2)\leq 2$ and $\chi(\cat_6)\leq 3$ in the main text. Using the fact that
\[
\langle 0|_m ~|\cat_m\rangle \propto |\cat_{m-1}\rangle
\]
we see that $\chi(\cat_{m})\geq \chi(\cat_{m-1})$ for all $m$. Therefore $\chi(\cat_5)\leq 3$ as well. We have $\chi(\cat_3)\leq\chi(\cat_4)\leq 2$ due to the decomposition
\[
|\cat_4\rangle=i|E\rangle+\left(\frac{1-i}{2}\right)\frac{1}{\sqrt{2}}\left(|0^4\rangle-i|1^4\rangle\right) \qquad \text{where} \qquad |E\rangle=\frac{1}{\sqrt{2^3}}\sum_{x\in \{0,1\}^4: |x| even} |x\rangle.
\]
Finally, we have $\chi(\cat_7)\leq \chi(\cat_8)\leq 6$ which follows from
\[
 \langle \cat_2|_{4,5}~|\cat_4\rangle|\cat_6\rangle\propto |\cat_8\rangle.
\]
The left hand side is easily seen to have rank at most $6$ since $\cat_2$ is a stabilizer state and $\chi(\cat_4)\leq 2$ and $\chi(\cat_6)\leq 3$.

Next let us discuss the matching lower bounds for $2\leq m\leq 6$. Clearly there is nothing to show for $m=2$ since all nonzero states have stabilizer rank at least $1$.

To establish the lower bounds for $m=3,\ldots , 6$ it will be convenient to define the \textit{Pauli spectrum} \cite{beverland2020lower} of an $m$-qubit state $\psi$ as the multiset
\[
\mathrm{PS}(\psi)=\left\{\langle \psi|P|\psi\rangle: P\in \mathcal{P}_m\right\}
\]
where $\mathcal{P}_m$ is the set of $m$-qubit Pauli operators. Note that all elements of the Pauli spectrum of a stabilizer state are either $+1,-1,0$. Moreover, we have $\mathrm{PS}(C|\psi\rangle)=\mathrm{PS}(|\psi\rangle)$ for any Clifford unitary $C$.

To show that $\chi(\cat_4)\geq \chi(\cat_3)\geq 2$ we need only show that $\cat_3$ is not a stabilizer state. To see this we note that $PS(\cat_3)$ contains elements which are not $+1,-1,0$. In particular,  $\langle \cat_3| X\otimes X\otimes I|\cat_3\rangle=1/2$ and therefore $\ket{\cat_3}$ is not a stabilizer state.

In the remainder of this appendix we show that $\chi(\cat_6)\geq \chi(\cat_5)\geq 3$. We use the fact that rank-2 states can be brought into a canonical form using the following lemma.

\begin{lem}
Let $\phi_1,\phi_2$ be two $n$-qubit stabilizer states. There exists a Clifford unitary $C$ such that 
\begin{equation}
C|\phi_1\rangle=|0^n\rangle \qquad \text{and} \qquad C|\phi_2\rangle\propto |1^a0^{b} +^{n-a-b}\rangle,
\label{eq:cprop}
\end{equation}
for some $a\in \{0,1\}$ and $0\leq b\leq n$.
\label{lem:gen}
\end{lem}
\begin{proof}
Let $D$ be a Clifford such that $D|\phi_1\rangle=|0^n\rangle$ and consider the stabilizer state $D|\phi_2\rangle$. Let $H$ be the single-qubit Hadamard on qubit $j\in [n]$ and write $H(y)=\prod_{j=1}^{n} (H_j)^{y_j}$ for any binary string $y\in \{0,1\}^n$.  Since $D|\phi_2\rangle$ is a stabilizer state it can be written in the so-called CH-form \cite{bravyi2018simulation} as
\begin{equation}
D|\phi_2\rangle \propto U_c H(r)|s\rangle
\label{eq:ch}
\end{equation}
for some $r,s\in \{0,1\}^n$ and some Clifford $U_c$ such that $U_c|0^n\rangle=|0^n\rangle$. 
We can furthermore find a permutation of the qubits $P$ such that
$$
P H(r) |s\rangle = \ket{x} \otimes H^{|r|} \ket{y}
$$
for some binary strings $x\in \{0,1\}^{n-|r|}, y\in \{0,1\}^{|r|}$.

First consider the case in which $\langle \phi_1|\phi_2\rangle\neq 0$.
Then $x=0$, and one can directly confirm that $C = (I\otimes Z(y))P(U_c)^{\dagger} D $ satisfies Eq.~\eqref{eq:cprop} with $a=0$ and $b=n-|r|$.

Next suppose  $\langle \phi_1|\phi_2\rangle=0$. In this case $x\neq 0$ and therefore we may choose the permutation $P$ such that $x=1x'$ for some $(n-|r|-1)$-bit string $x'$. We then choose $C=(V\otimes Z(y))P(U_c)^{\dagger} D$ where $V$ is a product of $CNOT$ gates controlled on the first bit of $x$ such that $V|x\rangle=|10^{n-|r|-1}\rangle$. Again we can directly confirm Eq.~\eqref{eq:cprop} with $a=1$ and $b=n-|r|-1$.
\end{proof}

\begin{lem}
$\chi(\mathrm{cat}_5) > 2$ .
\end{lem}
\begin{proof}
Below we shall use the following fact about the Pauli spectrum of $\cat_5$:
\begin{equation}
|\langle \cat_5|P|\cat_5\rangle|\in \{0,1/4,1/2,1\} \quad \text{for all } \quad P\in \mathcal{P}_5
\label{eq:catp}
\end{equation}
Since there are Paulis for which $|\langle \cat_5|P|\cat_5\rangle|\notin\{0,\pm 1\}$, we know that $\cat_5$ is not a stabilizer state.

Suppose that $\cat_5$ can be expressed as a superposition of two stabilizer states. Then by Lemma \ref{lem:gen} there is a Clifford $C$ such that 
\[
C|\cat_5\rangle \propto |0^n\rangle+\gamma |1^a0^b+^{n-a-b}\rangle
\]
for some $a\in \{0,1\}$, $0\leq b\leq n$, and $\gamma\in \mathbb{C}$. The fact that $\ket{\cat_5}$ is not a stabilizer state implies $|\gamma|>0$.
We can furthermore assume that $|\gamma| \leq 1$, if necessary by applying a Clifford that permutes $\ket{0^n}$ and $\ket{1^a 0^b +^{n-a-b}}$ then dividing by $|\gamma|$.

First let us suppose (to reach a contradiction) that $\cat_5$ can be written as a superposition of two orthogonal stabilizer states, i.e., the case $a=1$. Note that 
\[
2^b=\left| \{P\in \mathcal{P}_n: P|\cat_5\rangle=|\cat_5\rangle\}\right|,
\]
and therefore $b=1$. Thus we may find a Clifford $C$ such that
\begin{equation}
C|\cat_5\rangle =\frac{1}{\sqrt{1+|\gamma|^2}}\left( |0^5\rangle+\gamma |10+^{3}\rangle\right).
\label{eq:orthcase}
\end{equation}
Let us write $\gamma=re^{i\theta}$ with $0<r\leq 1$ and $\theta\in [0,2\pi)$. Now we see that 
\begin{align*}
\left|\langle \cat_5|C^{\dagger}Z_1C|\cat_5\rangle\right|&=\left|\frac{1-r^2}{1+r^2}\right| \\
\left|\langle \cat_5|C^{\dagger}X_1C|\cat_5\rangle\right|&=\left|\frac{r\cos(\theta)}{\sqrt{2}(1+r^2)}\right| \\
 \left|\langle \cat_5|C^{\dagger}X_1Y_4C|\cat_5\rangle\right|&=\left|\frac{r\sin(\theta)}{\sqrt{2}(1+r^2)}\right|.
\end{align*}
Recall from Eq.~\eqref{eq:catp} that each of these three Pauli expectation values must be in the set $\{0,1/4,1/2, 1\}$. Solving for $r, \theta$ which satisfy this constraint we get
\[
\gamma\in \{ \frac{1+i}{\sqrt{2}}, \frac{1-i}{\sqrt{2}}, \frac{-1-i}{\sqrt{2}},\frac{-1+i}{\sqrt{2}}\}
\]
For each of the four states Eq.~\eqref{eq:orthcase} corresponding to the above choices of $\gamma$ we computed the Pauli spectrum using a computer and compared with that of $|\mathrm{cat}_5\rangle$. We found that these multisets were different (in particular, the number of Paulis with expected value equal to zero is $710$ for each of the four candidate states, but $782$ for $|\mathrm{cat}_5\rangle$) and therefore we conclude that $\cat_5$ cannot be written as a superposition of two orthogonal stabilizer states.

Next let us suppose  (to reach a contradiction) that $\cat_5$ can be written as a superposition of two nonorthogonal stabilizer states, i.e., the case $a=0$. Again we may take $b=1$, so 
\begin{equation}
C|\cat_5\rangle=\frac{1}{\left(1+|\gamma|^2+\mathrm{Re}(\gamma)/4\right)^{1/2}}(|0^5\rangle+\gamma |0+^{4})\rangle.
\label{eq:nonorth}
\end{equation}
By a direct calculation we have
\[
\left|\langle \cat_5|C^{\dagger}Z_2X_3C|\cat_5\rangle\right|=\left|\frac{r\cos(\theta)}{2}\left(1+r^2+\frac{r\cos(\theta)}{2}\right)^{-1}\right|
\]
and
\[
\left|\langle \cat_5|C^{\dagger}Y_2C|\cat_5\rangle\right|=\left|\frac{r\sin(\theta)}{2}\left(1+r^2+\frac{r\cos(\theta)}{2}\right)^{-1}\right|.
\]
Recall from Eq.~\eqref{eq:catp} that both of these quantities must take values in $\{0,1/4,1/2,1\}$. Enforcing this constraint and solving for $0<r\leq 1$ and $\theta\in [0,2\pi)$ we find $\gamma=re^{i\theta}$ must satisfy
\[
\gamma \in \{\pm i,  -\frac{1}{2}\}.
\]
Thus we have three candidate states Eq.~\eqref{eq:nonorth} corresponding to the above three choices of $\gamma$. For  each of them we computed the Pauli spectrum using a computer and compared with that of $|\mathrm{cat}_5\rangle$. We found that these multisets were different and therefore reach a contradiction. We conclude that $|\mathrm{cat}_5\rangle$ cannot be written as a superposition of two stabilizer states.
\end{proof}
\end{document}